\documentclass[twoside,12pt,a4paper]{article}
\usepackage[english]{babel}
\usepackage[utf8]{inputenc}
\usepackage[twoside,margin=1.2in]{geometry}
\usepackage{graphicx}
\usepackage{amsmath}
\usepackage{amsthm}
\usepackage[square,sort]{natbib}
\usepackage{framed, color}
\definecolor{shadecolor}{rgb}{0.9, 0.9, 0.9}
\newtheorem{axiom}{Axiom}[section]
\newtheorem*{theorem}{Theorem}
\newcommand{\bu}[1]{\mbox{$\mathbf{#1}$}}
\usepackage{xspace}
\newcommand{\R}{\textnormal{\sffamily\bfseries R}\xspace}
\usepackage{Sweave}
\begin{document}

\setkeys{Gin}{width=0.8\textwidth}

\title{Uncertainty analysis and composite hypothesis under the likelihood paradigm}
\author{Andr\'e Chalom \and Paulo In\'acio de Knegt L\'opez de Prado,\\ 
  Department of Ecology, Institute of Biosciences, University of São Paulo \\
  \texttt{andrechalom@gmail.com}
  }
\date{07/28/2015}
\maketitle

\begin{abstract}
  The correct use and interpretation of models depends on several steps, two of which being the
  calibration by parameter estimation and the analysis of uncertainty. In the biological literature,
  these steps are seldom discussed together, but they can be seen as fitting pieces of the same puzzle.
  In particular, analytical procedures for uncertainty estimation may be masking a high degree of 
  uncertainty coming from a model with a stable structure, but insufficient data. 
  
  Under a likelihoodist approach, the problem of uncertainty estimation is closely related to the problem
  of composite hypothesis. In this paper, we present a brief historical background on the statistical 
  school of Likelihoodism, and examine the complex relations between the law of likelihood and the 
  problem of composite hypothesis, together with the existing proposals for coping with it. Then, we
  propose a new integrative methodology for the uncertainty estimation of models using the information
  in the collected data. We argue that this methodology is intuitively appealing under a likelihood 
  paradigm. 
\end{abstract}

\newpage
\vfill
\noindent
``In an ideal world, all data would come from well-designed experiments
and would be sufficient to simultaneously estimate all parameters using rigorous
statistical procedures. The world is not ideal. One must often combine estimates
from different experiments, or supplement high-quality data (...) with uncertain
data, or even assumptions (...). Think carefully about whether your conclusions
may be artifacts of your assumptions or calculations, and document those assumptions
and calculations so that your reader can ask the same question, and then carry on.''
\begin{flushright}
(H. Caswell, Matrix Population Models)
\end{flushright}
\vfill
\indent
\newpage
\tableofcontents
\newpage

\section{Motivation}

The use of mathematical and computational models is now widespread in several fields of the biological 
sciences. However, several applications of these models lack one or more steps that should be taken in
order to correctly understand and interpret its results \citep{Bart95}. Though seldom discussed together,
five steps can be seen as fundamental in this process: {\em verification, validation, calibration, 
uncertainty analysis} and {\em sensitivity analysis}. 

The two first steps are very tightly correlated, both in historical terms and in the common practice,
so it's common to refer to the set of validation and verification as ``V \& V''. We can think of verification
as trying to ensure that the computer code is correctly implementing the desired model, while validation
is trying to ensure that the theoretical model is capable of reproduce the actual phenomena of interest.
Another way to see it is to think that validation is determining that the model is solving the right 
equations, while verification is determining that the model is solving the equations right \citep{ASC2010}.

The third step, calibration, consists in determining the right values for the input parameters of a model
in order to enhance the correspondence between a model run and one observed situation. An appropriate
calibration of the model is then essential for the model to be used in a predictive fashion \citep{ASC2010}.

The last steps consist in determining how much the variation of the input parameters is translated into
the total variation of the results, which is called uncertainty analysis, and how much of the variation
of the results can be ascribed to the variation of each individual parameter, which is called sensitivity
analysis \citep{Helton03, Helton05}. These two have been extensively reviewed in another text 
(\citep{Chalom12}), so we will just highlight the fact that these analyses, together with model calibration,
strongly depend on previously and correctly executing the model verification and validation.

During the development and analysis of mathematical models in ecology, it is common to have completely 
separated procedures for calibration by parameter estimation (which we will call CPE) and 
sensitivity and uncertainty analyses (SUA) of a model. These steps are carried out in different sections of
the research papers, discussed in different chapters of books (ie: \citep{Caswell89}), and sometimes are
even performed by different people. In matrix models for population projections, the results discussion
is focused on a single combination of parameters, considered to be an optimum for the CPE, and the
SUA is presented as a separate step, being secondary to the main results (see, for example, 
\citep{SilvaMatos99}). We have no guarantees that even the statistical schools of though considered for 
both procedures will be coherent.

Moreover, the parameter set used in the model calibration may have few to no relation to parameters that
have biological relevance. In the case of matrix models, the former may be parameters related to a time 
series, while the later might be the projection matrix entries. This way, the researcher responsible for
the calibration may see the system as individual counts over time, while the researcher responsible for 
the uncertainty analyses will only be interested in the vital rates, which are real valued. When developed
in such disparate worlds, the SUA will be unable to point out directions for the planning of new experiments
or to single out which parameters should be targeted for more sampling effort.

This lack of integration between these fundamental steps for model analysis may have its roots in the
development of the CPE and SUA theories. While the parameter estimation has always been seen as belonging
to a statistical theory, the sensitivity of biological models was developed as an analytical tool, based on
the linear expansion of functions around a privileged point. Depending on the nature of the experiments
designed to measure vital rates, there may be a plethora of methods that might be used to determine the
set of values that best represent the status of knowledge that we have about a given natural system - and
these methods may stem from frequentist, Bayesian or likelihood-based approaches to the statistical theory.
The analytical theory of SUA, on the other hand, is in great part due to the work of Hal 
Caswell \citep{Caswell89}, and proceeds by taking a model that is already optimally parametrized, and studying
the first order derivatives of the model answer in relation to each input parameter. This formulation 
disregards completely the quantity and quality of collected data, except for their average, and sensitivity
measures are taken exclusively over the variation of the model answer in response to infinitesimal changes.
Thus, models that are parametrized by data with large uncertainties will not present larger sensitivity 
indexes than models parametrized by data with great precision. Increasing the sampling effort to gather
data about the system will not affect the uncertainty values. We can conclude that the analytical formulation
of the SUA, ultimately, gives information about the structure of the model, but is silent about the more
embracing questions from the collection of field data to the formulation and execution of the model. This might
lead to a false sense of confidence being placed on studies that present a stable matrix structure, leading
to small sensitivity indexes, but use data with large uncertainties.

Other than this, questions about model validation are not usually addressed by literature in this field,
but they are also intimately connected to the relevant questions about uncertainty. We can
divide the uncertainty about a model in three main sources: structural uncertainty, parameter uncertainty
and stochastic uncertainty (see \citep{Marino08} for a review about the latter two). While most
of the SUA techniques are focused on the second component, the validation of a model may point out
how much confidence we can have about whether one model (or a set of models) is adequate at
representing the desired phenomena.

A stochastic formulation for a theory of global uncertainty and sensitivity analyses, as described in 
\citep{Chalom12}, and interpreted under the likelihood paradigm for statistics, is able to regard
in the same way all of the uncertainty sources described above, generating a complete picture of our
knowledge about a system. Although these formulations may converge for very simple systems, this is not
the general case. 

The main advantages of using a stochastic formulation are:

\begin{enumerate}
  \item The analytical procedures are local, thus responding only to what happens after infinitesimal
    perturbations, and depends on the model functions being ``well-behaved'' in the chosen neighborhood.
    The stochastic procedures are global and have no such limitation.
  \item Any stochastic approach allows for the information contained in the sample variability to be used
    to represent the uncertainty about the parameters.
  \item By accepting the Likelihood Principle, we know that all relevant information on the samples is
    contained in the associated likelihood function. This way, we can use the likelihood function alone
    to proceed in the analysis, while the analytical procedures at times will disregard information, and
    at times will lead to the false sensation of having more detailed information than the sample allows.
\end{enumerate}

\section{Historical background}\label{sec:likelihood}

The statistical use of likelihoods is very widespread, starting from Fisher's definition on 1922, and
spanning several schools of thought. 
On the one hand, the framework for Neyman-Pearson's hypothesis testing is built upon it; on the other hand,
it is a bridge between the {\em a priori} and {\em a posteriori} probabilities in Bayesian analysis. 
Starting on the decade of 1960, the use of likelihoods as a basis for statistical inference started to be
more widely advocated as an alternative to both frequentism and Bayesianism. The 1965 book by Ian Hacking 
about the logic of inference and the 1972 book by A.W.F. Edwards on likelihoods can be seen as the starting
points for this position to become an independent school of thought, sometimes called Likelihoodism.
More recently, this position will be strongly supported by the works of Richard Royall and Elliott Sober.

The likelihoodist school of thought is based on the conjunction of the likelihood principle, as demonstrated
by Allan Birnbaum, and the law of likelihood, as enunciated by Ian Hacking.

Birnbaum derived the likelihood principle in 1962, from two generally accepted principles\footnote{
This demonstration is controversial to this date. See \citep{Mayo10} for a recent critique and 
\citep{Gandenberger12} for a recent defense}: the sufficiency principle and the conditionality principle.
Informally, we can think about the two first affirming that ``data that does not aggregate new information
is irrelevant to the inference'', and ``experiments that could have been realized but were not are irrelevant
to the inference''. From this principles, Birnbaum deduces the likelihood principle, which can informally
be stated as ``experimental results which are possible but were not observed are irrelevant to the 
inference''. Or, in Birnbaum's more formal terms:

\begin{quote}
``The likelihood principle: If $E$ and $E\prime$ are any two experiments with the same parameter space,
 represented respectively by density functions $f(x, \theta)$ and $g(y, \theta)$; and if $x$ and $y$ are
 any respective outcomes determining the same likelihood function; then $Ev(E, x) = Ev(E\prime, y)$.
 That is, the evidential meaning of any outcome $x$ of any experiment $E$ is characterized fully by giving
 the likelihood function $cf(x, \theta)$ (which need be described only up to an arbitrary positive constant
 factor), without other reference to the structure of $E$.''
\citep{Birnbaum62}
\end{quote}

The law of likelihood was formulated by Hacking as:

\begin{quote}
``If hypothesis $A$ implies that the probability that a random variable $X$ takes the value $x$ 
is $p_A(x)$, while hypothesis $B$ implies that the probability is $p_B(x)$, then the observation $X=x$
is \textbf{evidence supporting $A$ over $B$} if and only if $p_A(x) > p_B(x)$, and the likelihood ratio, 
$p_A(x)/p_B(x)$, measures the strength of that evidence.''
\citep{Hacking65}\footnote{Emphasis added}
\end{quote}

The divergence between the likelihood school and frequentist school is due to the later rejecting the
likelihood principle (although it accepts the principles of sufficiency and conditionality). On the other 
hand, the Bayesian school generally disagrees with the Likelihoodists by the interpretation of the law of
likelihood.

The likelihood school of thought can be seen as an intellectual heir to the Neyman-Pearson hypothesis
testing framework, as its breaking point from Fisherian significance testing happens because of the
logical necessity of comparing competing hypothesis. Fisher's tests give a special privileged position
to the so-called {\em null hypothesis}, which Neyman and Pearson will see as incomplete\footnote{
For a more complete and extremely didactic insight into their view, see the original paper from 
\citep{Neyman1933}}. The data
gathered by one experiment may indicate a small or a large support for one given hypothesis, but this alone
should not be considered as evidence in favor of this hypothesis before the other alternatives have
been also scrutinized. In other words, the probability or improbability of a single hypothesis should not
be used to conclude about its veracity. As the fictional detective Sherlock Holmes, by Sir Arthur Conan
Doyle famously said, ``How often have I said to you that when you have eliminated the impossible, 
whatever remains, {\em however improbable}, must be the truth?''. 

However, while Neyman and Pearson will use the likelihood ratio from several hypothesis to guide the 
{\em behavior} of the scientist in accepting or rejecting hypothesis, Royall and Edwards will separate
the concepts of (1) the degree of certainty that we have about one hypothesis; (2) the strength of evidence
that some data confers to one hypothesis over the others and (3) the course of action that should be taken
after examining such evidences. While Neyman and Pearson conflate the strength of evidence (2) with the
course of action (3), and Bayesian analysis conflate the degree of certainty (1) with the strength of 
evidence (2), the proposers of Likelihoodism propose that the role of statistical inference is to provide
{\em only} the strength of evidence \citep{Royall97}. 
The degree of certainty and course of action should take into account
a multiplicity of other factors, and should not be part of the inference. One of the greatest precursors 
of the statistical thought, Marquis Pierre-Simon de Laplace, had already foreshadowed this position with
the problem of the number of judges necessary to convict or acquit a prisoner: this decision, to Laplace,
needs to take into account not only the probability of convicting an innocent of acquitting a guilty 
prisoner, but also whether the punishment will be a fine or the death sentence\citep{Laplace1814}.

Moreover, by dissociating the strength of evidence from the decision making, Royall and Edwards are asking
the scientist to reveal the untransformed likelihood ratios derived from its data, instead of transforming
it into p-values or {\em a posterioris}, so that her peers are able to assess clearly what is the presented
evidence. By retaking the scientist role as a subjective decision maker, this proposition converges to
Fisher's, who criticized Neyman-Pearson's plan as being useful for industrial processing, but useless
to the scientific community. The likelihoodist proposal thus echoes the words of Fisher, and also the words
of the Bayesian probabilist I. J. Good:

\begin{quote}
``We have the duty of formulating, of summarizing, and of communicating our conclusions, in intelligible form, in
recognition of the right of other free minds to utilize them in making their own decisions.'' 
\citep{Fisher1955}
\end{quote}

\begin{quote}
``If a Bayesian is a subjectivist, he will know that the initial probability density varies from person to person and
so he will see the value of graphing the likelihood function for communication.'' \citep{Good76}
\end{quote}

However, there is a very important difference between the Bayesian and the Likelihoodist paradigms:
the {\em a posteriori} probabilities, constructed by the multiplication of the likelihood by an {\em a priori}
probability under the Bayesian school, is a absolute quantity. It is thus possible to refer to the {\em
a posteriori} probability of event $A$, or the {\em a posteriori} probability of an event $B$. In contrast,
the Likelihoodist approach will see in the likelihood ratio the end point for the statistical analysis. 
This is why the law of likelihood is expressed over {\bf evidence supporting $A$ over $B$}: because the
likelihoodist approach does not allow for the transformation of the relational support of $A$ over $B$ to
be translated into non-relational quantities of support for each hypothesis\footnote{See the text by
\citep{Fitelson07} for more details on this matter, including Carnap's classification of confirmation}.

\section{The challenge for uncertainty estimation}
From this section, we will consider problems where $\bu{x}$ represent a vector of data obtained independently
from one or more experiments from a vector-valued random variable $\bu{X}$, such that $P(\bu{X}\!=\!\bu{x}) = f(\bu{x};\theta)$. The parameter vector $\theta$ is unknown and may assume values in $\Theta$. We will
refer to the likelihood of $\theta$ given the observed data $\bu{x}$ as 
$\mathcal{L} (\theta | \bu{x}) = f(\bu{x}; \theta)$. The likelihood ratio over the same 
data set can be abbreviated as
$L(\theta_1, \theta_2) = \frac{\mathcal{L}(\theta_1|\bu{x})} {\mathcal{L}(\theta_2|\bu{x})}$. 

In order to perform an uncertainty analyses for a mathematical model under the likelihood paradigm, the 
general question that we would like to answer can be written as: ``how much support does the collected data
provide for the concurrent hypothesis over the model results?''. For example, in a population growth model,
the question may become ``how much support the collected data provide for the hypothesis that
the population is growing at the rate of 1 individual per year, {\em versus} the hypothesis that it is stable
?''\footnote{As we will see, the question can become more general, such as ``how much support the collected
data provide for the hypothesis that the population is growing or stable {\em versus} the hypothesis that is
declining?''}

This question, formulated over the parameters of a probability distribution, should be 
answered by means of the likelihood function. It can be demonstrated that one-to-one functions over the
parameters preserve the likelihood function, such that if $\mathcal{L}(\theta|\bu{x})$ is the likelihood
of $\theta$ and $\phi = f(\theta)$ is given by a one-to-one function $f(\cdot)$, the 
likelihood of $\phi$ is given by 
$\mathcal{L}(\phi|\bu{x}) = \mathcal{L} \left(f(\theta)|\bu{x}\right)$ \citep{Edwards72}.

However, the application of a mathematical model is analogous to the application of a generic function, 
and the reasoning above does not hold for arbitrary many-to-one functions. Thus, defining a methodology
for uncertainty analyses of arbitrary models is closely related to the problem of composite statistical
hypotheses. However, the law of likelihood is expressed in terms of simple statistical hypothesis, and there
is no way of combining these values to form the ``likelihood of a composite hypothesis''. As Edwards writes:

\begin{quote}
``No special meaning attaches to any part of the area under a likelihood curve, or to the sum of the likelihoods of two or more
hypotheses (...).
Although the likelihood function, and hence the curve, has the mathematical form of a [known] distribution, it does not
represent a statistical distribution in any sense.'' \citep{Edwards72}
\end{quote}

How can we work with composite hypothesis then? The traditional answer from the proposers of Likelihoodism is
that we don't. Edwards puts off composite hypothesis as ``uninteresting to science'', while Royall sees
this restriction to simple hypothesis as an advantage, not a limitation, to the likelihood approach.
However, this is necessary if we want to derive an uncertainty estimation for model results.
Here, we will present three paths that may be pursued for this end:

\begin{enumerate}
	\item Formulate a new law of likelihood, which can be applied for composite hypothesis;\label{i:alt}
	\item Present a methodology based on a logical extension of the law of likelihood;\label{i:ext}
	\item Address the issue using a model similar to the one applied for {\em nuisance parameters}. \label{i:prof}
\end{enumerate}

Proposal \ref{i:alt} is based on the axiomatic formulation of a generalized law of likelihood (GLL) that
is compatible with the law of likelihood (LL) for simple hypothesis for simple hypothesis, but able to extend
its domain for composite hypothesis. And proposal \ref{i:ext} is based on the use of a weaker definition
of statistical evidence, based on a weak law of likelihood (WLL). We will review some existing results related
to these proposals in the following sections.

Proposal \ref{i:prof} is based on methods such as profiled likelihood, which are {\em ad hoc} methods
widely used to reduce the study of multiparametric statistical models to visualizing one parameter at a time.
A typical use for these methods is exemplified by the problem of fitting a normal distribution over some data.
While, strictly, the comparable hypothesis must be given by pairs $(\mu = \mu_0, \sigma = \sigma_0)$,
representing one value for the mean and on for the variance of the distribution, it is possible to
compare the approximate support for different values for the mean, while leaving the variance free to assume
any viable value - thus treating the composite hypothesis $(\mu = \mu_0, \sigma \geq 0)$. We will get back
to this proposal in section \ref{sec:plue}.

\section{GLL: A general law of likelihood}\label{sec:GLL}

The most straightforward way to pursue a general law of likelihood is by considering the maximum\footnote{
Or more precisely, the {\em supremum}, as the set of likelihoods does not need to have a maximum} of the
likelihoods of the simple hypothesis as the likelihood of the composite hypothesis. This is precisely the
path followed by \cite{Zhang09, Zhang13} and \cite{Bickel10}. We will follow the more didactic approach by
Zhang. Let us consider two hypothesis, $H_1 : \theta \in \Theta_1 \subset \Theta$ versus 
$H_2 : \theta \in \Theta_2 \subset \Theta$. Now we postulate the following two axioms\footnote{The notation
was sightly altered for coherency with the other sections}:

\begin{axiom}
	If $\inf \mathcal{L}(\Theta_1 | \bu{x}) > \sup \mathcal{L} (\Theta_2 | \bu{x})$, then the observation
  $\bu{x}$ is evidence supporting $\Theta_1$ over $\Theta_2$.\label{ax:inf}
\end{axiom}

\begin{axiom}
	If $\bu{x}$ is evidence supporting $H_1^*$ over $H_2$ and $H_1^*$ implies $H_1$, then 
  $\bu{x}$ is evidence supporting $H_1$ over $H_2$.\label{ax:coh}
\end{axiom}

The first axiom ensures that if the images of $\mathcal{L} (\Theta_1|\bu{x})$ and 
$\mathcal{L} (\Theta_2|\bu{x})$ are disjoint intervals, or in other words, if every simple hypothesis
that implies $H_1$ is supported over every simple hypothesis that implies $H_2$, then $H_1$ is supported
over $H_2$. There seem to be no reason to reject this axiom, other than a predisposition to ignore composite
hypothesis. The second axiom is used to ensure a form of logical coherence. It is important to stress that
this coherence axiom is {\em not} used to suppose that the logic structure of the hypothesis is somehow
preserved in the likelihood function. In particular, if $H_1^*$ implies $H_1$, the axiom \ref{ax:coh} does
not warrant the assertion that $H_1^*$ is better supported than $H_1$.

From these axioms, we can derive the following general law of likelihood (GLL):

\begin{theorem}
	\textbf{(GLL)} ``If $\sup \mathcal{L} (\Theta_1 | \bu{x} ) > \sup \mathcal{L} (\Theta_2 | \bu{x})$, then
  there $\bu{x}$ is evidence supporting $H_1$ over $H_2$.'' \citep{Zhang09}
\end{theorem}

\begin{proof}
	Consider $\sup \mathcal{L} (\Theta_1 | \bu{x} ) > \sup \mathcal{L} (\Theta_2 | \bu{x})$. 
  Then there exists $\theta_1 \in \Theta_1$
	so that $\mathcal{L} (\theta_1 | \bu{x}) > \sup \mathcal{L} (\Theta_2 | \bu{x})$. 
  From the axiom \ref{ax:inf}, the hypothesis
	$H_1^* : \theta = \theta_1$ is supported over $H_2$. But $H_1^*$ implies $H_1$, so the conclusion
  follows directly from axiom \ref{ax:coh}.
\end{proof}

The use of a ``generalized likelihood ratio'' given by $\sup \mathcal{L}(\Theta_1 | \bu{x}) / 
\sup \mathcal{L} (\Theta_2 | \bu{x}) $ seems natural to quantify the strength of evidence, just as the
likelihood ratio is used to quantify the strength of evidence under the law of likelihood (LL).
It is trivial to see that the GLL is compatible with LL in the case of simple hypothesis, and with some
particular cases proposed by \citep{Royall97}. However, it is striking to see that this formulation of the 
GLL allows for an {\em absolute} degree of hypothesis confirmations. One hypothesis 
$H_1: \theta \in \Theta_1$ can be absolutely supported if $\sup \mathcal{L}(\Theta_1|\bu{x}) > 
\sup \mathcal{L} (\Theta_1^c|\bu{x}$, where $^c$ represents the complimentary set.
That is equivalent to $L(\Theta_1, \Theta_1^c) > 1$. This absolute likelihood will always be zero
for simple hypothesis concerning real values parameters, but this does not happen then $\Theta$ is finite.

On the other hand, accepting the generalized likelihood ratio makes it impossible for a researcher
to combine evidence from multiple data sources, as the generalized likelihood is not a multiplicative
quantity. So, the greatest problem here is to find a generalization that allows for a measure of
strength of evidence that can be combined across experiments.

\section{WLL: A weak law of likelihood}

A more general way of expressing statistical support is given by 

\begin{description}
\item[WLL]
``Evidence $E$ favors hypothesis $H_1$ over hypothesis $H_2$ {\em if} $P(E|H_1) > P(E|H_2)$ {\em and}
$P(E|\sim\!\!H_1) \leq P(E|\sim\!\!H_2)$.''
\citep{Fitelson07}
\end{description}

Evidently, LL $\implies$ WLL. However, WLL is also implied by most modern Bayesian proponents. Given a 
confirmation measure $c(H,E)$, the Bayesian equivalent for the LL is:

\begin{description}
\item[$\dagger_c$]
``Evidence $E$ favors hypothesis $H_1$ over hypothesis $H_2$ according to measure $c$ if and only if
  $c(H_1,E) > c(H_2,E)$.''
\citep{Fitelson07}
\end{description}

There are three main formats for the confirmation measures:
\begin{description}
\item[Difference] $d(H,E) = P(H|E)-P(H)$
\item[Ratio] $r(H,E) = \frac{P(H|E)}{P(H)}$
\item[Likelihood ratio] $l(H,E) = \frac{P(E|H)}{P(E|\sim\!\!H)}$
\end{description}

It is important to remember that the Bayesian confirmation given by $c(H,E)$ has a non-relational nature,
and the relation confirmation is constructed by comparing two of those measures. Likelihoodists will see
in the LL a primitive relational relation. While the construction of the non-relational confirmation measures
depends on the specification of {\em a prioris}, the WLL makes no such assumptions, using only {\em catch-all}
probabilities (the $P(E|\sim\!\!H)$ terms). One theory based on the WLL without invoking {\em a prioris}
could be used as a base to an intermediate inference school of thought, one that is not dependent on the
specification of {\em a prioris}, but without the arbitrary restriction on the use of composite hypothesis.

\section{PLUE: a proposal for likelihood profiling}\label{sec:plue}

In this section, we will describe how the profiling techniques for likelihood functions can be 
seen as a foundation for a methodology of uncertainty estimation for mathematical models.
We reason that this methodology is intuitively appealing under a likelihood paradigm; and we will
refer to it as PLUE - Profiled Likelihood Uncertainty Estimation.

\subsection{Intuition}

Suppose that we have one mathematical model, for example a structured population growth model
like those exemplified in \citep{Chalom12}, and one data set from which we can estimate the 
vital rates of the population, such as survival, growth and fertility rates.

With this data in hands, we would like to ask the following questions:

{\em ``What is the support given by the data to the affirmation that the population is stable {\em versus}
that it is declining? What is the support given by the data to the affirmation that the population will
become extinct in up to 10 years {\em versus} more than 10 years?''}

These questions cannot be answered under a strict Likelihoodist point of view, as they refer to
composite hypothesis over the parameters. However, we can take the maximum likelihood point as 
privileged and restrict the question to the form:

{\em ``What is the support given by the data to point of maximum likelihood {\em versus} any point in which
the population growth is negative?''}

Schematically, we can see in figure \ref{fig:esquema} that we are comparing the likelihood at the global
maximum ($x_{max}$) with the maximum value attained at a distinct region ($x_{lim}$). This kind of question
is now comparing one simple hypothesis, corresponding to the maximum likelihood, with a composite hypothesis.
This kind of comparison avoids the problems and quirks that afflict the general laws of likelihood due to
the superposition of likelihood intervals. In terms of the Zhang axioms (see section \ref{sec:GLL}),
we are accepting a weak form of the first, but not the second.

This reasoning can be extended to the following:

{\em ``Which are the points in the parameter space for which the support given by the data 
{\em versus} the point of maximum likelihood is greater than a certain threshold $\delta$?''}

By asking the same question for different values of $\delta$, we are effectively profiling the 
likelihood surface for the parameter space. In this point, we need to remember that profiling 
methods for likelihood are used in inference since the decade of 1970 to reduce the dimensionality
of parameters spaces, in particular in the form of elimination of ``{\em nuisance parameters}''.
Even if the status of this procedures is not fully understood under the likelihood
paradigm (see for example \citep{Kalbfleisch70} for the relation between nuisance parameters and fiducial
inference), these methods are generally accepted as a valid procedure for exploratory analysis.

We have already discussed the typical example of fitting a normal distribution to a data set. In
this example, the likelihood is defined for pairs $(\mu = \mu_0, \sigma = \sigma_0)$,
but it is usual to make affirmations about the average $\mu$ of the distribution, without 
referring to $\sigma$.
However, this is equivalent to consider the model $g(\mu, \sigma) = \mu$, that is, a model that
projects the bi-dimensional space formed by $\mu$ and $\sigma$ into an unidimensional space. Our proposal
is merely to use this same reasoning for general non-invertible functions. Even if this is nothing
new under the optics of the likelihood profiling, this is the first time that this methodology is proposed
in the literature of uncertainty analyses, to the best of our knowledge.

\begin{figure}[htb]
\includegraphics{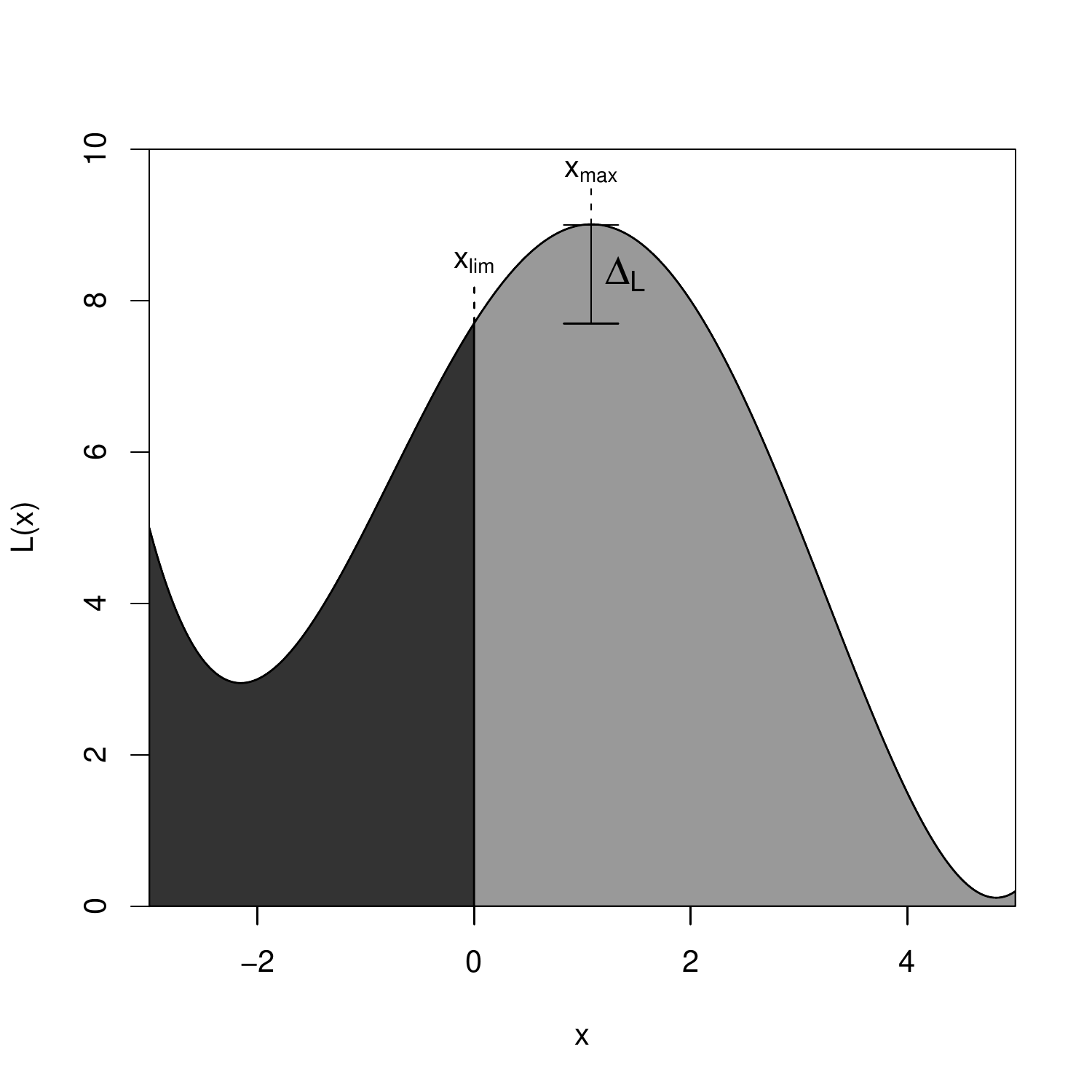}
	\caption{Schematic representation of a likelihood function, highlighting two interest points: the
  global maximum $x_{max}$ and the maximum in a given region $x_{lim}$. See details on the main text.} 
	\label{fig:esquema}
\end{figure}

\subsection{Method}
\begin{description}
	\item[Definitions]
    Consider first a mathematical model of interest. We will call this the ``biological model''\footnote{
    But it could be a physical model, hydrological model, geochemical model, etc.}
    to differentiate between this and the statistical model presented below. We will consider the simple
    case of a scalar response $y$ for a model with a vector input $\bu{\theta}$: 
		$y = F(\bu{\theta})$.

    Let's represent by $\chi$ the set of data collected in field or laboratory experiments. If we have
    $n$ parameters and $m$ observations, $\chi$ is a table with $n$ columns and $m$ rows.

    Now formulate different models to explain your data (ex: constant or grouped fertility, constant or
    decreasing growth rate, models with 4 or 5 size classes, etc.). These will be the statistical models.
    Write the likelihood function $\mathcal{L}(\bu{\theta}|\chi)$ for each model, find the best
    set of parameters that fits your data and determine the value of AIC for each statistical model
    (see for example \citep{Burnham02} for details in this step). 
	\item[Sampling]
    Take the model with the best AIC, and use a Monte Carlo sampling to generate a large number of 
    samples with density proportional to 
    $\mathcal{L}(\bu{\theta}|\chi)$ (for example, use the Metropolis algorithm \citep{Tierney94}). 
    We will call this sample $\bu{A}$. To each element in 
    $A_{i \cdot}$ (taken from $\mathcal{L}(\bu{\theta}|\chi)$), we have associated the value of 
		$L_i = \mathcal{L}(A_{i \cdot} | \chi)$, the likelihood of this sample, and we take
		$Y_i = F(A_{i \cdot})$, the result of the biological model over this sample.
		Normalize $L_i$ to have the minimum log-likelihood in $0$.
	\item[Aggregation]
    Now we take the model results $Y_i$ and associated likelihoods $L_i$, and create an upper profile
    in the following manner: for each increment $z$, find the largest value 
		$\bar y$ in $Y_i$ such that $L_i \leq z$. Note down this value as $P_{sup}(z) = \bar y$,
    and repeat for a larger value of $z$.

    Proceed in an analogous fashion to build the lower likelihood profile. Both profiles, taken together,
    can be used to investigate the plausibility regions of $y$ under $\mathcal{L}(\bu{\theta}|\chi)$.
\end{description}

\section[Case study]{Case study: a minimal model for population growth}

Here, we will study a structured population model that be considered minimal:
the life stages are non-reproducing juveniles and reproductive adults. This example 
is presented by \citep{Caswell08} as a simple example of the analytical sensitivity analysis. The model
is described by the following matrix:

\begin{equation}
 A = \left[
 \begin{array}{ll}
	 \sigma_1 (1-\gamma) &   f \\
     \sigma_1 \gamma & \sigma_2
 \end{array}
 \right]
\end{equation}

Here, $\sigma_1$ is the probability of juvenile survival, $\sigma_2$ is the probability of adult survival,
$\gamma$ is the maturation probability and $f$ is the adult fertility. We will represent by $\lambda$ the
largest eigenvalue of this matrix, and consider this as the model response.

To simplify some calculations, let's presume that the survival rate does not depend on the life stage
($\sigma_1=\sigma_2=\sigma$). We will also suppose that it is possible to unequivocally mark the juveniles
that were born in the last time interval, and which passed through the maturation process to become adults
in the last time interval. The model becomes:

\begin{equation}
 A = \left[
 \begin{array}{ll}
	 \sigma (1-\gamma) &   f \\
     \sigma \gamma & \sigma
 \end{array}
 \right]
\end{equation}

For large animals with one offspring for reproductive seasons, such as whales, $f$ can be modeled as the
proportion of adults that generate offspring. The parameter $\sigma$ is modeled as the proportion of 
individuals that survive from one time interval to another, and $\gamma$ as the proportion of juveniles that
maturate from one time interval to another. This way, all of the three parameters can be modeled as
binomial distributions with unknown probabilities 
$\theta_i$ and number of trials given by 
$n_1$, the number of juveniles, $n_2$, the number of adults, and $n_t$, the total population size:

\begin{align}
	\gamma & \sim \operatorname{binom}(\theta_1, n_1) \\
	f      & \sim \operatorname{binom}(\theta_2, n_2) \\
	\sigma & \sim \operatorname{binom}(\theta_3, n_t = n_1+n_2)
\end{align}

We will also assume that the parameters are independent, to arrive at the likelihood functions graphed in
\ref{fig:LikFunc} (mathematical details are presented in section \ref{apmat}).

We will examine a numerical example with the initial population having 10 juveniles and15
adults. The population size is very small, which is important to highlight the difference between the
analytical and stochastic approaches. After one time step, we can see 
3 recent adults, 2 born juveniles and 23 total survivors. It's 
easy to see from the figure \ref{fig:LikFunc} that the best estimate for this parameters is given by
$\sigma = $ 0.92, $f = $ 0.13 and $\gamma = $ 0.3.
Also, the corresponding value for $\lambda$ is 1.02.

\begin{figure}
	\caption{Likelihood function for each parameter in the model. Black = $\sigma$, red = $f$ and green = $\gamma$.}
	\label{fig:LikFunc}
\end{figure}

The likelihood functions for each parameter were used to generate
50000 samples using Metropolis algorithm, from which we derive an empirical distribution for
$\lambda$, which is proportional to likelihood of the parameters.
This $\lambda$ distribution, together with the likelihood values associated with each input, was used to
construct a likelihood profile for the model result. The minimum likelihood for 
$\lambda$ is attained in $\lambda = $ 1.02.

Figures \ref{fig:lambdascatter} and \ref{fig:lambdaprcc} show preliminary results of the application of
sensitivity techniques (such as described in \citep{Chalom12}) to the generated samples. It is important
to highlight that this analyses have been carried out in a neighborhood of the maximum likelihood point
that is not infinitesimal, nor arbitrary, which would be the case in analytical and other stochastic methods.

\begin{figure}
\includegraphics{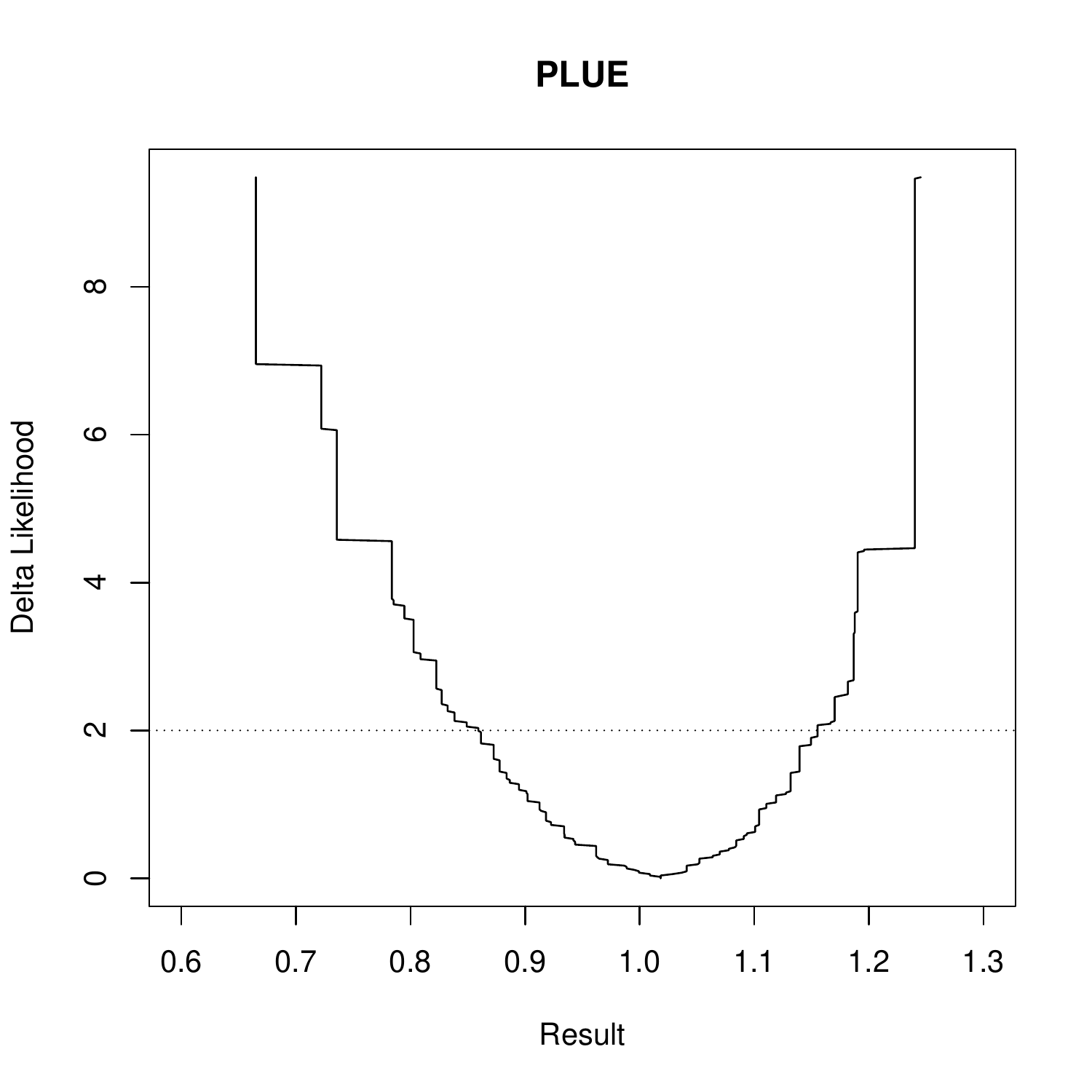}
	\caption{Likelihood profile for the results of a minimal model of structured population growth.}
	\label{fig:lambda}
\end{figure}
\begin{figure}
\includegraphics{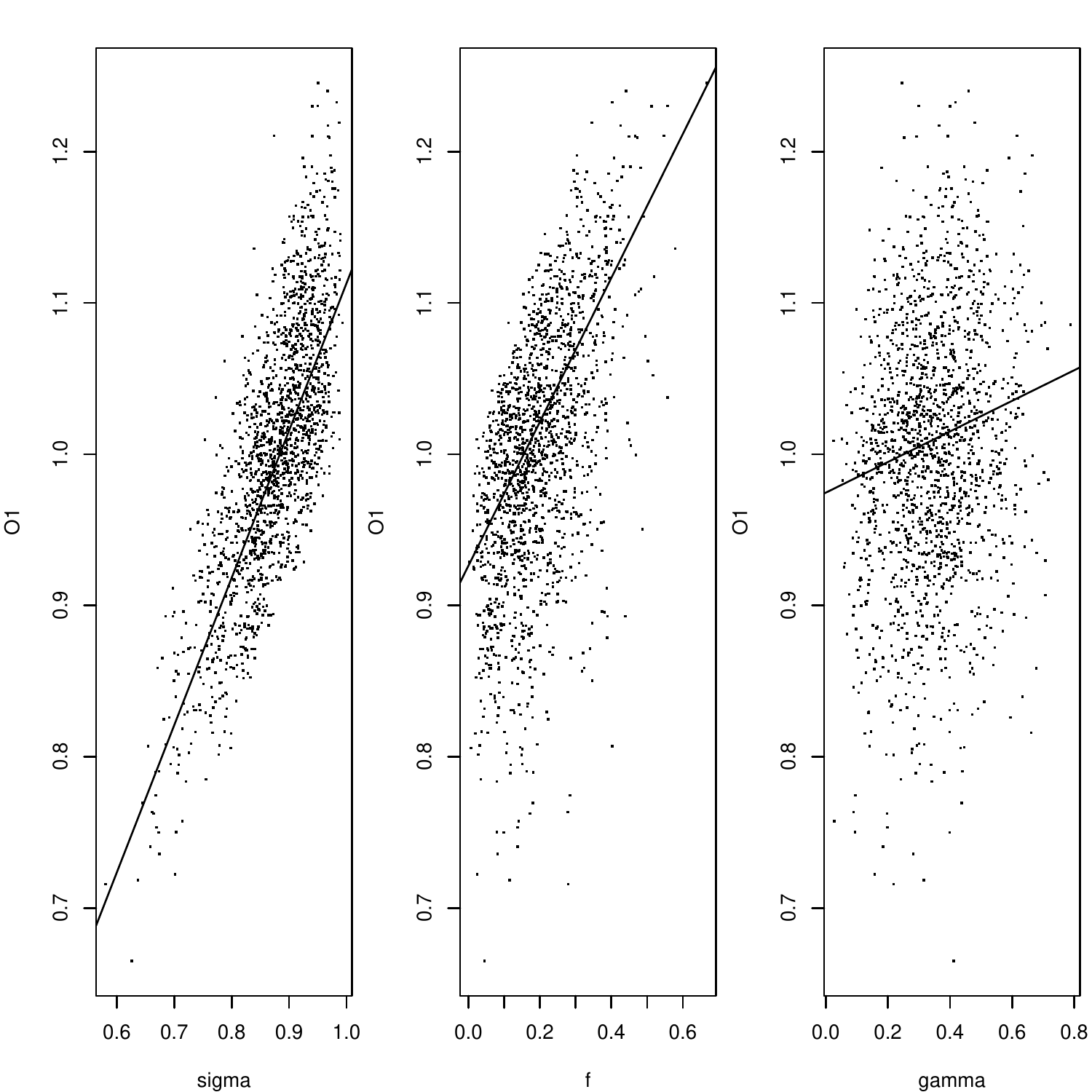}
	\caption{Graph of the dispersion of parameter values (in the x axis) and results of the 
  minimal model of structured population growth, generated by a likelihood profile method.  }
	\label{fig:lambdascatter}
\end{figure}
\begin{figure}
\includegraphics{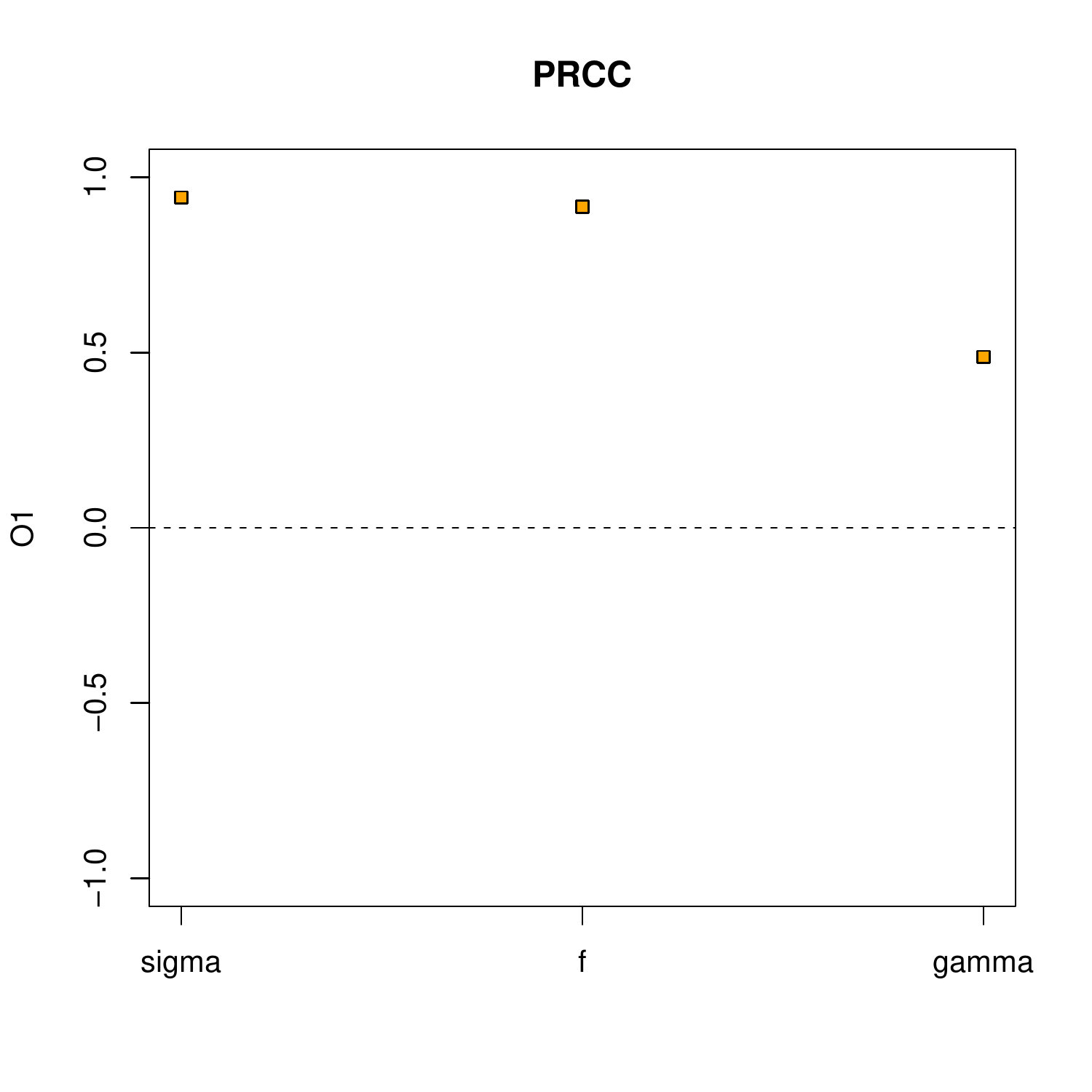}
	\caption{Partial Rank Correlation Coefficient (PRCC) between the input and output values for a 
  minimal model of structured population growth.}
	\label{fig:lambdaprcc}
\end{figure}

The analysis shown indicates that the estimated value of $\lambda$ is very unreliable, presenting a
very open profile. It is important to notice that this happens due to the small sample size. If we consider
a sample in which all the observations are 3 times larger, but keeping all the proportions,
the analysis results in a more closed profile (fig. \ref{fig:lambda2}).

\begin{figure}
\includegraphics{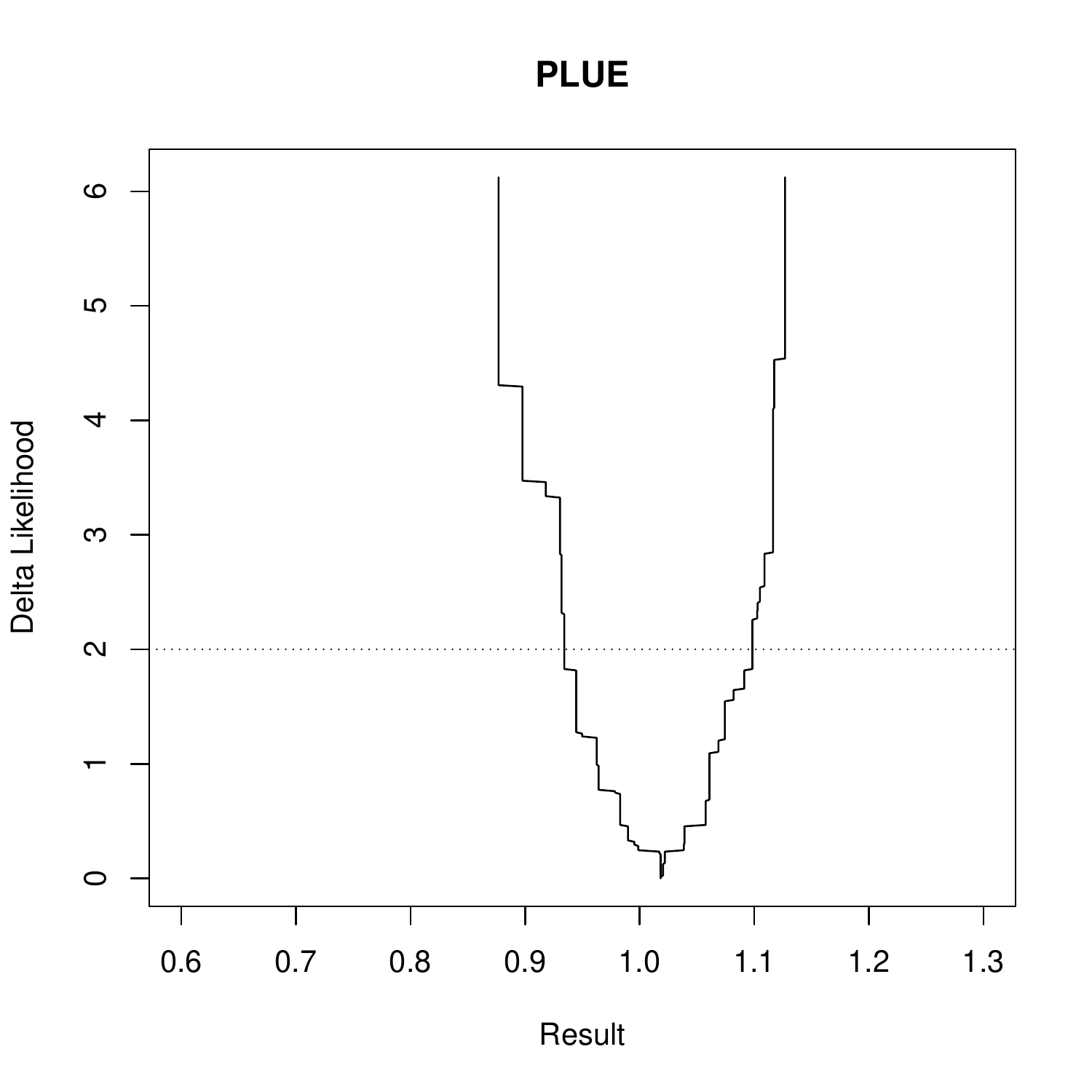}
	\caption{Likelihood profile for the results of a minimal model of structured population growth, but with
  larger sample size. In comparison to \ref{fig:lambda}, this profile is more closed and the plausibility
  regions are narrower. }
	\label{fig:lambda2}
\end{figure}

\newpage
\subsection{Mathematical details}\label{apmat}
In this section, we will detail some mathematical points left out in the example above.

For a given observed number of juveniles that have undergone maturation, becoming adults in the last time
interval, $x_A$, we can write the log-likelihood function for $\gamma$ as:

\begin{equation}
\mathcal{L} \left( \theta_1 | x_A \right) 
= \log \left( {n_1 \choose x_A} \theta_1^{x_A} (1-\theta_1) ^{n_1-x_A} \right)
\end{equation}

In the same way, the log-likelihood for $f$ is related to the number of juveniles that were born in the
last time interval, $x_J$, by the expression:

\begin{equation}
\mathcal{L} \left( \theta_2 | x_J \right) 
= \log \left( {n_2 \choose x_J} \theta_2^{x_J} (1-\theta_2) ^{n_2-x_J} \right)
\end{equation}

And last, the $\sigma$ log-likelihood function is given over the number of surviving individuals $x_S$
(or the total population minus $x_J$):
\begin{equation}
\mathcal{L} \left( \theta_3 | x_S \right) 
= \log \left( {n_t \choose x_S} \theta_3^{x_S} (1-\theta_3) ^{n_t-x_S} \right)
\end{equation}

With the strong assumption that the three variables are independent, we can write $n_t = n_3$ and
$\{x_A, x_J, x_S\} = \{x_1, x_2, x_3\}$ in order to arrive at the more economical notation\footnote{
It is important to remember that $x_1$ is not the number of juveniles, and so on.}
for the likelihood function over the entire parameter vector $\boldsymbol\theta$:

\begin{equation}
\mathcal{L} \left( \boldsymbol{\theta} | \mathbf{x} \right) 
= \sum_i \log \left( {n_i \choose x_i} \theta_i^{x_i} (1-\theta_i) ^{n_i-x_i} \right) \label{eqn:loglik}
\end{equation}

The result of the model is $\lambda$, the largest eigenvalue in $A$. For 2x2 matrices, this is given simply by:
\begin{equation}
	\det \left[ 
	\begin{array}{ll} \lambda - \sigma(1-\gamma) & -f \\
		-\sigma \gamma &         \lambda - d 
	\end{array}
	\right]
	= \lambda^2 - \lambda \operatorname{tr}(A) + \operatorname{det}(A) = 0
\end{equation}
\begin{align}
	\lambda = \frac{1}{2} \left(\operatorname{tr}(A) + \sqrt{\operatorname{tr}^2(A) - 4 \operatorname{det}(A)} \right)
\end{align}

\subsection{\R code for performing the analyses}

In this section, we will present the \R code used to generate the analyses and graphs presented above.
This can be seen as a tutorial in the usage of the computational tools provided in the ``pse'' package.
You should have installed \R 
\footnote{This tutorial was written and tested with \R version 
3.0.1, but it should work with newer versions}
along with an interface and 
text editor of your liking, and the package ``pse''
(available on http://cran.r-project.org/web/packages/pse).

\subsubsection{Biological and statistical models}
First, we should define our interest model. We will refer to this model as the biological model\footnote{
Because the models of interest in my research are biological. It can also be a physical model, 
geochemical model, etc.} to distinguish this from the statistical model we will be using to 
estimate likelihoods. This model 
must be formulated as an \R function that
receives a {\em data.frame}, in which every column represent a different
parameter, and every line represents a different combination of values
for those parameters. The function must return an array with the same
number of elements as there were lines in the original data frame,
and each entry in the array should correspond to the result of running
the model with the corresponding parameter combination. For example, it can be this\footnote{See the preceding
sections for the interpretation of the model, including parameters and results}:

\begin{Schunk}
\begin{Sinput}
> tr <- function (A) return(A[1,1]+A[2,2])
> A.to.lambda <- function(A) 
+   1/2*(tr(A) + sqrt((tr(A)^2 - 4*det(A))))
> getlambda = function (sigma, f, gamma) {
+   A.to.lambda (matrix(c(sigma*(1-gamma), f, 
+                         sigma*gamma, sigma), 
+                       ncol=2, byrow=TRUE) )
+ }
> getlambda = Vectorize(getlambda)
> model <- function(x) getlambda(x[,1], x[,2], x[,3])
\end{Sinput}
\end{Schunk}

Following the definition of the model, we should define the likelihood function for our parameters.
To do this, we can formulate and test several statistical models. In order to fit competing models to the 
data and select the best of them, we recommend using the \R package \textbf{bbmle}.
Then, the best model should be written as a function receiving a numeric vector representing one realization
of the parameter vector and returning the {\em positive} log-likelihood of that vector.

For example, the best model may be that the parameters $\sigma$, $f$ and $\gamma$ are all independent 
from each other, coming from the specified binomial distributions:

\begin{Schunk}
\begin{Sinput}
> # Initial population: juveniles, adults, total
> n <- c(10, 15); n.t <- sum(n)
> # Observed quantities: matured, born, survival
> obs <- c(3, 2, 23)
> # Best guess for the parameters
> sigma <- obs[3]/n.t
> f <- obs[2]/n[2]
> gamma <- obs[1]/n[1]
> lambda <- getlambda(sigma, f, gamma)
\end{Sinput}
\end{Schunk}

The likelihood function, in this case, should be:

\begin{Schunk}
\begin{Sinput}
> # NOTE: LL function uses GLOBAL obs and n!!!
> LL <- function (x) 
+ {
+   t <- dbinom(obs[3], n.t, as.numeric(x[1]), log=TRUE) +
+     dbinom(obs[2], n[2], as.numeric(x[2]), log=TRUE) +
+     dbinom(obs[1], n[1], as.numeric(x[3]), log=TRUE)
+   if (is.nan(t)) return (-Inf);
+   return(t);
+ }
\end{Sinput}
\end{Schunk}

Please note that this function uses the global variable {\em obs}, and that it return minus infinity instead
of not-a-number in cases where the likelihood is not properly defined. This can happen, for instance, if
any of the values of {\em x} is negative.

\subsubsection{Profiling: sampling and aggregating the results}
After carefully constructing the model of interest and the likelihood function, as described in the 
previous section, performing the PLUE analysis is simply a matter of calling the {\em PLUE} function.
This function performs three steps. First, it performs a Monte Carlo sampling of the likelihood function
in order to generate a large sample from the likelihood distribution. Then, the biological model is applied
to this sample, and finally the model results are combined by means of profiling the likelihood function
associated with each data point.

The {\em pse} package implements a simple Metropolis sampling function that can be used
by setting {\em method=`internal'} in the {\em PLUE} function call. For more elaborate sampling schemes, 
and more control over the process, we recommend using {\em method=`mcmc'}, which uses the {\em mcmc} \R 
package.

\begin{Schunk}
\begin{Sinput}
> library(pse)
> factors = c("sigma", "f", "gamma")
> N = 50000
> # we set the random seed to ensure reproducibility:
> set.seed(42)
> # The starting point for the Monte Carlo sampling
> start = c(sigma, f, gamma)
> plue <- PLUE(model, factors, N, LL, start, 
+              method="mcmc", opts=list(blen=10), nboot=50)
\end{Sinput}
\end{Schunk}

In order to see the profiled likelihood of the model result, simply run:
\begin{Schunk}
\begin{Sinput}
> plot(plue, xlim=c(0.6, 1.3))
\end{Sinput}
\end{Schunk}

Notice that additional parameters may be specified for the underlying plotting function, such as
{\em xlim}, {\em ylab} or {\em main}. Additional plots may be seen by using the 
{\em plotscatter} and {\em plotprcc} functions:
\begin{Schunk}
\begin{Sinput}
> # Sensitivity analyses over lambda
> plotscatter(plue)
> plotprcc(plue)
\end{Sinput}
\end{Schunk}

The interpretation of these graphs is analogous to those generated by the Latin Hypercube
Sampling, as described in \citep{Chalom12}. However, it is important to notice that, instead of the 
arbitrary region of the parameter space that is sampled in the LHS scheme, the plots 
presented in this document are representing a discretization of the likelihood surfaces
of the parameters, thus incorporating all the information about the data collected.

\section*{Acknowledgments}
This work was supported by a CAPES scholarship.

\newpage
\bibliographystyle{apalike}
\bibliography{chalom}

\begin{thebibliography}{}

\bibitem[Bart, 1995]{Bart95}
Bart, J. (1995).
\newblock Acceptance criteria for using indivual-based models to make
  management decisions.
\newblock {\em Ecological applications}, 5(2):411--420.

\bibitem[Bickel, 2010]{Bickel10}
Bickel, D. (2010).
\newblock The strength of statistical evidence for composite hypotheses:
  Inference to the best explanation.
\newblock {COBRA} Preprint Series.

\bibitem[Birnbaum, 1962]{Birnbaum62}
Birnbaum, A. (1962).
\newblock On the foundations of statistical inference.
\newblock {\em Journal of the American Statistical Association},
  57(298):269--326.

\bibitem[Burnham and Anderson, 2002]{Burnham02}
Burnham, K. and Anderson, D. (2002).
\newblock {\em Model Selection and Multimodel Inference: A Practical
  Information-Theoretic Approach}.
\newblock Springer.

\bibitem[Caswell, 1989]{Caswell89}
Caswell, H. (1989).
\newblock {\em Matrix population models}.
\newblock John Wiley \& Sons.

\bibitem[Caswell, 2008]{Caswell08}
Caswell, H. (2008).
\newblock Perturbation analysis of nonlinear matrix population models.
\newblock {\em Demographic Research}, 18:59--116.

\bibitem[Chalom and Prado, 2012]{Chalom12}
Chalom, A. and Prado, P. (2012).
\newblock Parameter space exploration of ecological models.
\newblock arXiv:1210.6278 [q-bio.QM].

\bibitem[Edwards, 1972]{Edwards72}
Edwards, A. (1972).
\newblock {\em Likelihood}.
\newblock Cambridge University Press, Cambridge.

\bibitem[Fisher, 1955]{Fisher1955}
Fisher, R. (1955).
\newblock Statistical methods and scientific induction.
\newblock {\em Philosophical Transactions of the Royal Society of London.
  Series B}, 17:69--78.

\bibitem[Fitelson, 2007]{Fitelson07}
Fitelson, B. (2007).
\newblock Likelihoodism, {B}ayesianism, and relational confirmation.
\newblock {\em Synthese}, 156(3):473--489.

\bibitem[Gandenberger, 2012]{Gandenberger12}
Gandenberger, G. (2012).
\newblock A new proof of the likelihood principle.
\newblock {\em British Journal for the Philosophy of Science}.

\bibitem[Good, 1976]{Good76}
Good, I. (1976).
\newblock The {B}ayesian influence, or how to sweep subjectivism under the
  carpet.
\newblock In Harper, W. and Hooker, C., editors, {\em Foundations of
  probability theory, statistical inference, and statistical theories of
  science}, page~0. Reidel.

\bibitem[Hacking, 1965]{Hacking65}
Hacking, I. (1965).
\newblock {\em Logic of statistical inference}.
\newblock Cambridge university press, New York.

\bibitem[Helton et~al., 2005]{Helton05}
Helton, J., Davis, F., and Johnson, J. (2005).
\newblock A comparison of uncertainty and sensivity analysis results obtained
  with random and latin hypercube sampling.
\newblock {\em Reliability Engineering and System Safety}, 89:304--330.

\bibitem[Helton and Davis, 2003]{Helton03}
Helton, J. and Davis, J. (2003).
\newblock Latin hypercube sampling and the propagation of uncertainty in
  analyses of complex systems.
\newblock {\em Reliability Engineering and System Safety}, 81:23--69.

\bibitem[Kalbfleisch and Sprott, 1970]{Kalbfleisch70}
Kalbfleisch, J.~D. and Sprott, D.~A. (1970).
\newblock Application of likelihood methods to models involving large numbers
  of parameters.
\newblock {\em Journal of the Royal Statistical Society. Series B
  (Methodological)}, pages 175--208.

\bibitem[Laplace, 1814]{Laplace1814}
Laplace, P. (1814).
\newblock {\em Essai philosophique sur les probabilit\'es}.
\newblock Bachelier, imprimeur-libraire, Paris.

\bibitem[Marino et~al., 2008]{Marino08}
Marino, S., Hogue, I., Ray, C., and Kirschner, D. (2008).
\newblock A methodology for performing global uncertainty and sensivity
  analysis in systems biology.
\newblock {\em Journal of Theoretical Biology}, 254:178--196.

\bibitem[Mayo, 2010]{Mayo10}
Mayo, D. (2010).
\newblock An error in the argument from conditionality and sufficiency to the
  likelihood principle.
\newblock In Mayo, D. and Spanos, A., editors, {\em Error and Inference: Recent
  Exchanges on Experimental Reasoning, Reliability and the Objectivity and
  Rationality of Science}, pages 305--314. Cambridge University Press.

\bibitem[Meisner, 2010]{ASC2010}
Meisner, R. (2010).
\newblock {Advanced Simulation and Computing Program Plan FY11}.
\newblock Technical report, {Office of Advanced Simulation and Computing, NNSA
  Defense Programs}.
\newblock {NA-ASC-122R-10}.

\bibitem[Neyman and Pearson, 1933]{Neyman1933}
Neyman, J. and Pearson, E. (1933).
\newblock On the problem of the most efficient tests of statistical hypotheses.
\newblock {\em Philosophical Transactions of the Royal Society of London.
  Series A}, 231:289--337.

\bibitem[Royall, 1997]{Royall97}
Royall, R. (1997).
\newblock {\em Statistical Evidence: A Likelihood Paradigm}.
\newblock {Chapman and Hall}, London.

\bibitem[Silva~Matos et~al., 1999]{SilvaMatos99}
Silva~Matos, D., Freckleton, R., and Watkinson, A. (1999).
\newblock The role of density dependence in the population dynamics of a
  tropical palm.
\newblock {\em Ecology}, 80(8):2635--2650.

\bibitem[Tierney, 1994]{Tierney94}
Tierney, L. (1994).
\newblock Markov chains for exploring posterior distributions (with
  discussion).
\newblock {\em Annals of Statistics}, pages 1701--1762.

\bibitem[Zhang, 2009]{Zhang09}
Zhang, Z. (2009).
\newblock A law of likelihood for composite hypotheses.
\newblock arXiv:0901.0463 [math.ST].

\bibitem[Zhang and Zhang, 2013]{Zhang13}
Zhang, Z. and Zhang, B. (2013).
\newblock A likelihood paradigm for clinical trials.
\newblock {\em Journal of Statistical Theory and Practice}, 7(2):157--177.

\end{thebibliography}
\end{document}